\documentclass[submission,copyright]{eptcs}
\providecommand{\event}{LFMTP 2011} 
\usepackage{breakurl}             

\usepackage{palatino}
\usepackage{float}
\floatstyle{boxed} \restylefloat{figure}

\usepackage{amsfonts}
\usepackage{amsmath}
\usepackage{amsthm}
\usepackage[english]{babel}
\hypersetup{colorlinks,urlcolor=black}
\usepackage{prooftree}
\usepackage{stmaryrd}
\usepackage{amssymb} 
\usepackage{graphicx}
\usepackage{mdwlist}

\usepackage{array}
\newcolumntype{L}[1]{>{$}p{#1}<{$}}
\newcolumntype{C}[1]{>{\centering$}p{#1}<{$}}
\newcolumntype{R}[1]{>{\raggedleft$}p{#1}<{$}}

 \renewenvironment{thebibliography}[1]{%
   \begin{oldthebibliography}{#1}%
     \setlength{\parskip}{0ex}%
     \setlength{\itemsep}{0ex}%
     \fontsize{9.5}{8.5} 
     \selectfont
}%
 {%
   \end{oldthebibliography}%
 }

\newtheoremstyle{jamiestyle}
  {4pt}
  {0pt}
  {\it}
  {0pt}
  {\bf}
  {.}
  { }
  {}
\theoremstyle{jamiestyle}
\newtheorem{thrm}{Theorem}[section]
\newtheorem{prop}[thrm]{Proposition}
\newtheorem{lemm}[thrm]{Lemma}
\newtheorem{corr}[thrm]{Corollary}
\newtheoremstyle{jamienfstyle}
  {4pt}
  {0pt}
  {\normalfont}
  {0pt}
  {\bf}
  {.}
  { }
  {}
\theoremstyle{jamienfstyle}
\newtheorem{nttn}[thrm]{Notation}
\newtheorem{defn}[thrm]{Definition}
\newtheorem{xmpl}[thrm]{Example}
\newtheorem{rmrk}[thrm]{Remark}

\usepackage{fancybox}
\newlength{\mylength}
\newenvironment{frameqn}%
{\setlength{\fboxsep}{6pt}
\setlength{\mylength}{\linewidth}
\addtolength{\mylength}{-2\fboxsep}%
\addtolength{\mylength}{-2\fboxrule}%
\Sbox
\minipage{\mylength}%
\setlength{\abovedisplayskip}{0pt}%
\setlength{\belowdisplayskip}{0pt}%
$$}%
{$$\endminipage\endSbox
\[\fbox{\TheSbox}\]}

{\setlength{\fboxsep}{5pt}
\setlength{\mylength}{\textwidth}%
\addtolength{\mylength}{-2\fboxsep}%
\addtolength{\mylength}{-2\fboxrule}%
\Sbox
\minipage{\mylength}%
\setlength{\abovedisplayskip}{0pt}%
\setlength{\belowdisplayskip}{0pt}%
}%
{\endminipage\endSbox
\[\fbox{\TheSbox}\]}

\newcommand\vect[1]{\overline{#1}}
\newcommand\somerel{\mathrel{\mathcal R}}
\newcommand{\rtm}{r}
\newcommand{\stm}{s}
\newcommand{\ttm}{t}

\newcommand{\aeq}{=_{\scriptstyle\alpha}}
\newcommand{\beq}{=_{\scriptstyle{\beta}}}

\newcommand{\deffont}[1]{\textbf{#1}}

\newcommand{\f}[1]{\ensuremath{\text{$\mathit{#1}$}}}
\newcommand{\fcomp}{\circ}

\newcommand{\lam}[1]{\lambda{#1}.}
\newcommand{\beeq}[1]{\ensuremath{\langle #1 \rangle_{\scalebox{.5}{\hspace{-.2ex}$\beta$}}}} 

\newcommand{\rulefont}[1]{\ensuremath{(\mathbf{#1})}}
\newcommand{\sm}{\mapsto}

\newcommand{\ssm}{:=}
\newcommand{\dom}{\f{dom}}
\newcommand{\supp}{\f{supp}}

\newcommand{\theory}[1]{\ensuremath{\mathsf{#1}}}

\newcommand\fto{{\to}}
\newcommand\limp{{\Rightarrow}}
\newcommand\liff{\Leftrightarrow}

\newcommand{\fa}{\f{fa}}

\newcommand{\fv}{\f{fv}}

\newcommand\fix{\f{fix}}
\usepackage{marvosym}

\newcommand\nontriv{\f{nontriv}}

\newcommand\den[1]{{\hspace{.00ex}\scalebox{.55}{$#1$}}}
\newcommand\iden{\den{\interp I}} 
\newcommand\jden{\den{\interp J}} 

\newcommand{\idenot}[2]{\denot{\interp I}{#1}{#2}}
\newcommand{\jdenot}[2]{\denot{\interp J}{#1}{#2}}
\usepackage[mathscr]{eucal}
\newcommand\interp[1]{\ensuremath{\mathscr #1}}

\newcommand{\denot}[3]{\llbracket #3 \rrbracket_{\scalebox{.6}{$#2$}}^\den{#1}} 

\newcommand{\act}{{\cdot}}

\newcommand\cent{\vdash}
\newcommand\ment{\vDash}

\newcommand\gloment{\vDash_{\scalebox{.4}{glo}}}

\newcommand{\id}{{id}}

\newcommand\atomsdown{\mathbb A^{\hspace{-.3ex}{}_{{}^<}}}
\newcommand\atomsup{\mathbb A^{\hspace{-.3ex}{}_{{}^>}\!}} 

\newcommand\Forall[1]{\forall #1.}
\newcommand\Exists[1]{\exists #1.}

\title{Nominal Henkin Semantics: simply-typed lambda-calculus models in nominal sets}
\date{}
\author{Murdoch J. Gabbay
\and
Dominic P. Mulligan\thanks{\tiny 
We are very grateful to Peter Selinger and to two anonymous referees for their useful and constructive comments.
We acknowledge the support of the Leverhulme trust.
We acknowledge the support of grant RYC-2006-002131 at the Polytechnic University of Madrid. 
The project CerCo acknowledges the financial support of the Future and Emerging Technologies (FET) programme within the Seventh Framework Programme for Research of the European Commission, under FET-Open grant number: 243881.  
}}
\def\titlerunning{Nominal Henkin Semantics}
\def\authorrunning{Gabbay \& Mulligan}

\begin{document}
\maketitle
\begin{abstract}
We investigate a class of nominal algebraic Henkin-style models for the simply typed $\lambda$-calculus in which variables map to names in the denotation and $\lambda$-abstraction maps to a (non-functional) name-abstraction operation.
The resulting denotations are smaller and better-behaved, in ways we make precise, than functional valuation-based models.

Using these new models, we then develop a generalisation of $\lambda$-term syntax enriching them with existential meta-variables, thus yielding a theory of incomplete functions.
This incompleteness is orthogonal to the usual notion of incompleteness given by function abstraction and application, and corresponds to holes and incomplete objects. 
\end{abstract}


\section{Introduction}
\label{sect.introduction}

In this paper we develop a Henkin-style semantics for the simply-typed $\lambda$-calculus in nominal sets.
The simply-typed $\lambda$-calculus (STLC) has notions of typed variable, substitution, and function abstraction.
Correspondingly, our models in nominal sets will enrich `ordinary' sets with typed names, a substitution action, and name-abstraction.
Thus, concepts that are normally characteristic of syntax---like variable, substitution, and variable-binding---are explicitly represented as nominal algebraic structure \cite{gabbay:nomuae}.

The resulting models have different properties from traditional valuation-style (`closed') semantics.
Intuitively this is because leaving names in the denotation gives the models more structure---we have more information about `where an element came from'.

For instance, Proposition~\ref{prop.xi} (the \rulefont{\xi} rule) and Theorem~\ref{thrm.well-pointed} (well-pointedness) are properties that hold of the nominal models of this paper, and fail for `classical' treatments (see Examples~\ref{xmpl.two.point} and~\ref{xmpl.one.point}). 
This is because the direct inclusion of names into the denotation forces there to be `enough' elements of the model, and naturality requirements of the models require these elements to be `sufficiently distinguishable'.
These conditions cannot be expressed without names in the denotation.

Furthermore, we find that we can extend this to a syntax and semantics for existential variables. 
That is, we will extend STLC syntax with `holes'. 
The technique used is essentially the same as the \emph{nominal terms} of \cite{gabbay:nomu-jv} (a permissive variant thereof, following \cite{gabbay:perntu,gabbay:perntu-jv}) but taking semantics in nominal models of STLC instead of in datatypes of abstract syntax with binding.

%
%

Because $\lambda$-abstraction maps to atoms-abstraction, the denotation of functions does not involve function spaces. 
Because variables map to themselves, valuations are not used either; their role is taken by the substitution for names.
Thus we obtain a simple `first-order flavoured' completeness proof (Theorem~\ref{thrm.st.complete}).

In summary, nominal Henkin models differ from `ordinary' Henkin models by including variables and substitution in the underlying domain of the denotation as nominal algebraic structure.
This yields a new class of models which seems to not display certain pathologies of the `ordinary' models, and which can be leveraged to design novel calculi with applications e.g. to existential variables.

\section{Background}

\subsubsection*{Background on simply-typed $\lambda$-calculus}

\begin{defn}
\label{defn.atoms}
Fix a countably infinite set of \deffont{atoms} $\mathbb A$.

We use a \deffont{permutative} convention that $a,b,c,\ldots$ range over \emph{distinct} atoms (so for instance in Definition~\ref{defn.sub} the $a_i$ are silently assumed distinct, in Definition~\ref{def.permutation} $a$ and $b$ are taken distinct, and so on).
\end{defn}

\begin{defn}
\label{defn.pts.sorts}
\begin{enumerate*}
\item
Fix a nonempty set of \deffont{base types} $\tau\in\f{BaseTypes}$.
Define \deffont{(simple) types} by
$\phi ::= \tau \mid \phi\to\phi$.
Let $\phi$, $\psi$, $\chi$ range over types.
\label{defn.terms.and.substitutions}
\item
Fix a set of \deffont{constants} $C \in \f{Constants}$, to each of which is associated a type $\f{type}(C)$.

\noindent Define \deffont{terms} by:
$\rtm      ::=  a \mid C \mid \lam{a{:}\phi}\rtm \mid \rtm\rtm$.
Let $r$, $s$, $t$ range over terms.
$\lambda a$ binds $a$ in $\lam{a{:}\phi}\stm$ and we take terms up to $\alpha$-equivalence as usual. 
\end{enumerate*}
\noindent Define \deffont{free atoms} $\fa(\rtm)$ by
$\fa(a) =  \{ a \}$,\ 
$\fa(C) = \varnothing$,\ 
$\fa(rs) = \fa(r)\cup\fa(s)$,\ 
and
$\fa(\lam{a{:}\phi}r) =  \fa(r) {\setminus} \{ a \}$. 
\end{defn}

\begin{defn}
\label{defn.sub}
Give terms a \deffont{capture-avoiding substitution action} $\rtm[a_i{\ssm}s_i]_1^n$ (side-conditions can be guaranteed by $\alpha$-renaming):
$$
\hspace{-.5em}\begin{array}{r@{\ }l@{\quad}r@{\ }l@{\quad}l}
a_j[a_i{\ssm}s_i]_1^n=&s_j
&
b[a_i{\ssm}s_i]_1^n=&b
\\
C[a_i{\ssm}s_i]_1^n=&C
&
(\lam{c{:}\chi}r)[a_i{\ssm}s_i]_1^n=&\lam{c{:}\chi}(r[a_i{\ssm}s_i]_1^n) &(c\not\in\bigcup_1^n(\{a_i\}\cup\fa(s_i)))
\\
(rs)[a_i{\ssm}s_i]_1^n=&(r[a_i{\ssm}s_i]_1^n)(s[a_i{\ssm}s_i]_1^n)
\end{array}
$$
\end{defn}


\begin{defn}
\label{defn.congruence}
\label{defn.beta.equivalence}
Let $\beq$ be the least equivalence on terms (up to $\alpha$-equivalence) such that: 
\begin{displaymath}
\begin{prooftree}
\rtm\beq\rtm'\quad \stm\beq\stm'
\justifies
\rtm\stm\beq\rtm'\stm
\using\rulefont{CongApp}
\end{prooftree}
\qquad
\begin{prooftree}
\stm\beq\stm'
\justifies
\lam{a{:}\phi}\stm\beq\lam{a{:}\phi}\stm'
\using\rulefont{\xi}
\end{prooftree}
\qquad
\begin{prooftree}
\phantom{}
\justifies
(\lam{a{:}\phi}\rtm) \ttm \beq \rtm[a \ssm \ttm]
\using\rulefont{\beta}
\end{prooftree}
\end{displaymath}
\end{defn}


\begin{defn}
\label{defn.st.environments}
A \deffont{type environment} $\Gamma$ is a set of \deffont{atomic typings} $a:\phi$ which is \emph{functional} in the sense that if $a:\phi$ and $a:\phi'$ then $\phi=\phi'$.

\deffont{Derivable} typing judgements $\Gamma\cent r:\phi$ are defined using the (standard) rules in Figure~\ref{fig.typing.rules}.

Define the \deffont{domain} of $\Gamma$ by $\dom(\Gamma)=\{a\mid \Exists{\phi}(a{:}\phi\in\Gamma)\}$.
Write $\Gamma,a{:}\phi$ for the type environment obtained by adding $a{:}\phi$ to $\Gamma$; if we write this, we impose a condition that $a\not\in\dom(\Gamma)$. 
\end{defn}


\begin{figure}
\begin{displaymath}
\begin{array}{c}
\begin{prooftree}
\phantom{h} 
\justifies
\Gamma,a{:}\phi \cent a : \phi
\using\rulefont{V}
\end{prooftree}
\quad
\begin{prooftree}
(\f{type}(C)=\phi)
\justifies
\Gamma \cent C : \phi
\using\rulefont{C}
\end{prooftree}
\quad
\begin{prooftree}
\Gamma,a{:}\phi \cent \rtm:\psi\ \ (a{\not\in}\dom(\Gamma)) 
\justifies
\Gamma\cent (\lam{a{:}\phi}\rtm) : \phi\fto\psi 
\using\rulefont{L}
\end{prooftree}
\quad
\begin{prooftree}
\Gamma \cent \rtm : \phi\fto\psi 
\ \ 
\Gamma \cent \stm : \phi 
\justifies
\Gamma \cent \rtm\stm : \psi
\using\rulefont{A}
\end{prooftree}
\end{array}
\end{displaymath}
\caption{Typing rules for the simply-typed $\lambda$-calculus (STLC)}
\label{fig.typing.rules}
\end{figure}

\begin{defn}
\label{defn.typing.rules}
A \deffont{typing judgement} is a tuple $\Gamma\cent r:\phi$.
The \deffont{derivable} typing judgements are defined in Figure~\ref{fig.typing.rules}.
\end{defn}

\subsubsection*{Background on nominal sets}



\begin{defn}
\label{defn.U}
The cumulative hierarchy of \deffont{ZFA sets} $\mathcal U$ is the least fixed point of $\mathcal U=\mathbb A\cup\f{powerset}(\mathcal U)$.
This can be constructed by starting from atoms and transfinitely adding all subsets (a construction going back to Von Neumann \cite{vonneumann:ubewam}).
\end{defn}

\begin{defn}
\label{defn.swap}
\label{def.permutation}
\label{def.nontriv}
Given $a,b\in\mathbb A$ write $(a\ b)$ for the \deffont{swapping} bijection on atoms mapping $a$ to $b$, $b$ to $a$, and any other $c\in\mathbb A\setminus\{a,b\}$ to $c$.

If $\pi$ is a bijection on atoms define $\nontriv(\pi)=\{a\mid \pi(a)\neq a\}$.
 
Write $\mathbb P$ for the group of bijections (finitely) generated by swappings, and call these bijections \deffont{permutations}.

Write $\pi \fcomp \pi'$ for the \deffont{composition} of $\pi$ and $\pi'$ (so $(\pi\circ\pi')(a)=\pi(\pi'(a))$).
Write $\id$ for the \deffont{identity} permutation (so $\id(a)=a$ always). 
\end{defn}

\begin{lemm}
\label{lemm.basic.property.of.permutations}
A bijection $\pi$ on atoms is a permutation if and only if $\nontriv(\pi) = \{a \mid \pi(a) \neq a \}$ is finite.
\end{lemm}

\begin{defn}
\label{defn.perm}
Give $\mathcal U$ a \deffont{permutation action} $\pi\act x$ inductively defined by $\pi\act a=\pi(a)$ for $a\in\mathbb A$, and
$\pi\act X=\{\pi\act x\mid x\in X\}$ for $X\in\mathcal U\setminus\mathbb A$.

If $A\subseteq\mathbb A$ write $\fix(A)=\{\pi\in\mathbb P\mid \Forall{a{\in}A}\pi(a)=a\}$.

Say that $A\subseteq\mathbb A$ \deffont{supports} $x\in\mathcal U$ when $\Forall{\pi{\in}\fix(A)}\pi\act x=x$.
\end{defn}

\begin{defn}
Call an element $x\in\mathcal U$ \deffont{finitely-supported} when it has a unique least finite supporting set $\supp(x)$.
Write $a\#x$ for $a\not\in\supp(x)$ and read this as `$a$ is \deffont{fresh for} $x$'.
\end{defn}

\begin{lemm}[\cite{gabbay:newaas-jv,gabbay:fountl}]
If $x\in\mathcal U$ has a finite supporting set $A$, then $\supp(x)$ exists.
\end{lemm} 



Our reasoning can be formalised in first-order logic, enriched with the axioms of Zermelo-Fraenkel set theory with atoms (\deffont{ZFA}).
This is just a formal way of stating that we have assumed atoms and sets and we reason about them mathematically, but stating it in terms of formal logic lets us express an important observation, that our reasoning is symmetric under permutation: 
\begin{thrm}
\label{thrm.equivar}
If $\vect x$ denotes a list $x_1,\ldots,x_n$, write
$\pi\act \vect x$ for $\pi\act x_1,\ldots,\pi\act x_n$.
Suppose $\Phi(\vect x)$ is a ZFA predicate on variables included in $\vect x$.
Then we have \emph{equivariance} \cite[Section~4]{gabbay:fountl}:\footnote{$\vect x$ must contain \emph{all} the variables mentioned in the predicate.
It is not the case that $a=a$ if and only if $a=b$---but it is the case that $a=b$ if and only if $b=a$.} 
$
\Phi(\vect x) \liff \Phi(\pi\act \vect x) .
$
\end{thrm}
We will appeal to equivariance repeatedly to quickly yet rigorously rename atoms, usually while retaining an inductive hypothesis.
See for instance Lemma~\ref{lemm.st.beta.1}.\footnote{This technique was used in pencil-and-paper mathematics instead of long inductive proofs, e.g. in \cite{gabbay:frelog}.}




\begin{defn}
\label{defn.triv}
Say that $X\in\mathcal U\setminus\mathbb A$ has the \deffont{trivial} action when $\supp(x)=\varnothing$ for every $x\in X$ (equivalently: when $\pi\act x=x$ for every $x\in X$ and permutation $\pi$).
\end{defn}

\begin{defn}
If $X,Y\in\mathcal U\setminus\mathbb A$ then a \deffont{function(-set)} from $X$ to $Y$ is a subset of $X\times Y$ such that $\Forall{x{\in}X}\Exists{y{\in}Y}(x,y)\in f$ and $\Forall{x,y,y'}((x,y)\in f\land (x,y')\in f)\limp y=y'$.
Write $X\fto Y$ for the set of all functions from $X$ to $Y$.
Write $X\limp Y$ for the set of all functions from $X$ to $Y$ with finite support.
\end{defn}

\begin{rmrk}
The permutation action from Definition~\ref{defn.perm} gives $f\in X\fto Y$ the \emph{conjugation} permutation action specified by $\pi\act (f(x)) = (\pi\act f)(\pi\act x)$.
\end{rmrk}

\begin{lemm}
\label{lemm.triv.func}
If $X$ and $Y$ in $\mathcal U\setminus\mathbb A$ have the trivial permutation action (Definition~\ref{defn.triv}), then so does $X\fto Y$, and $X\fto Y= X\limp Y$. 
(If underlying sets have empty support then so do functions between them.) 
\end{lemm}

\section{Nominal models for simple type theory}
\label{sect.nominal.models}
\subsection{Nominal $\lambda$-model}
\label{subsect.lambda.model}

\begin{nttn}
\label{nttn.pi.Gamma}
Write $\pi\act\Gamma=\{\pi(a){:}\phi\mid a{:}\phi\in\Gamma\}$. 
\end{nttn}

\begin{defn}
\label{defn.lambda.model}
A \deffont{model} $\interp I$ consists of an assignment for each type environment $\Gamma$ and type $\phi$ of a finitely-supported set $\idenot{\Gamma}{\phi}$ together with the following data:
\begin{enumerate*}
\item
For every $a{:}\phi\in\Gamma$ an element $a^\iden_\phi\in\idenot{\Gamma}{\phi}$.
\item
For every constant $C$ an element $C^\iden\in\idenot{\Gamma}{\f{type}(C)}$.
\item
If $x\in\idenot{\Gamma,a{:}\phi}{\psi}$,\ an element $[a{:}\phi]x\in\idenot{\Gamma}{\phi\fto\psi}$.
\item
For $x\in\idenot{\Gamma}{\phi\fto\psi}$ and $y\in\idenot{\Gamma}{\phi}$,\ an element $x\bullet y\in\idenot{\Gamma}{\psi}$.
\item
$\idenot{\Gamma\cap\Gamma'}{\phi}=\idenot{\Gamma}{\phi}\cap\idenot{\Gamma'}{\phi}$.
\item
\label{item.subset.condition}
If $x\in\idenot{\Gamma}{\phi}$ then $\supp(x)\subseteq\dom(\Gamma)$.
\end{enumerate*}
$\interp I$ must be \emph{equivariant} in the sense that:
$$
\begin{gathered}
\pi\act\idenot{\Gamma}{\phi}=\{\pi\act x\mid x\in\idenot{\Gamma}{\phi}\}=\idenot{\pi\act\Gamma}{\phi}
\qquad
\pi\act a^\iden_\phi=(\pi(a))^\iden_\phi
\qquad
\pi\act C^\iden =C^\iden
\\
\pi\act [a{:}\phi]x=[\pi(a){:}\phi]\pi\act x
\qquad
\pi\act (x\bullet y)=(\pi\act x)\bullet(\pi\act y)
\end{gathered}
$$
We write $x[a\sm y]$ as sugar for $([a{:}\phi]x)\bullet y$. 
In addition, $\interp I$ must be a \deffont{nominal algebra for substitution} by satisfying rules \rulefont{Suba}, \rulefont{Sub\#}, \rulefont{SubApp}, and \rulefont{Sub\text{$\lambda$}}; we fill in types as appropriate (we discuss \rulefont{SubId} below):
$$
\begin{array}{l@{\quad\qquad}l@{\ }l@{\ =\ }l}
\rulefont{Suba} &&a^\iden_\phi[a\sm x] & x
\\
\rulefont{Sub\#} & a\#z\limp& z[a\sm x] & z
\\
\rulefont{SubApp} & &(z'\bullet z)[a\sm x] & (z'[a\sm x])\bullet (z[a\sm x])
\\
\rulefont{Sub\text{$\lambda$}} & c\#x\limp &([c{:}\chi]z)[a\sm x] & [c{:}\chi](z[a\sm x])
\\
\rulefont{SubId} &&z[a\sm a^\iden] & z
\end{array}
$$
\end{defn}

For the rest of this subsection fix a model $\interp I$. 

Let us break down the design of Definition~\ref{defn.lambda.model}.
Obviously names inhabit the denotation in a very direct and literal sense that $a^\iden_\phi\in\idenot{\Gamma}{\phi}$.
The reader can think of $a:\phi$ as a constant which must be interpreted `as itself' by $a^\iden_\phi$.

But $a^\iden_\phi$ also behaves like a variable:
It can be renamed by $\pi\act x$, and bound by $[a{:}\phi]x$, and it can also be substituted for.
The rules \rulefont{Suba} to \rulefont{Sub\text{$\lambda$}} do the job that valuations do in `normal' models; they replace a name $a^\iden_\phi$ by an(other) element of the model.
The significant difference is that in standard models we pick a valuation and then form a denotation; in nominal models we form a denotation and then---if we wish---substitute for the free variables.

The axioms \rulefont{Suba} to \rulefont{SubId} can be made formal in nominal algebra \cite{gabbay:nomuae}.
These particular axioms are taken from \cite{gabbay:capasn-jv}.\footnote{\rulefont{Suba} to \rulefont{SubId} soundly and completely characterise the \emph{syntactic} model of substitution.
In this paper we are also interested in non-syntactic models, so weaker axioms---and thus more models---are reasonable.  
We chose the axioms above because they are \emph{closed}, in the sense of \cite{gabbay:nomr-jv,gabbay:clonre}, which gives better computational properties (if we ever design an abstract machine using this semantics).
}
Instead of \rulefont{Sub\#} we could take a weaker axiom $b[a{\sm} x]=b$.
Conversely, we could safely add \rulefont{SubId} and thus exclude certain arguably pathological models.
The language of Definition~\ref{defn.pts.sorts} is not expressive enough to detect these choices, but the language of Definition~\ref{defn.holes.sorts} is (see Example~\ref{xmpl.hole.power}).


Aside from the inclusion of names, our notion of model resembles Henkin models, which have an applicative structure in which abstractions have a well-defined interpretation \cite{henkin:comtot}.

Just as is the case for Henkin models, Definition~\ref{defn.lambda.model} specifies what a model must look like but does not build one.
We do build a concrete model out of syntax as part of the completeness proof in Subsection~\ref{subsect.completeness}.

The equivariance conditions are standard for nominal techniques; our models must be symmetric up to permuting atoms.

Finally, conditions~1 to~6 specify the structure of a model that makes it into a model of the $\lambda$-calculus, by interpreting names (as themselves), constants, $\lambda$-abstraction (as a function of the name $a$ and the element $x$)\footnote{So $[a{:}\phi]x$ need not be precisely equal to the Gabbay-Pitts atoms-abstraction $[a]x$ from \cite{gabbay:newaas-jv}.} and a Henkin-models style application. 

\begin{defn}
\label{defn.interp}
Suppose $\Gamma\cent r:\phi$.
Define an \deffont{interpretation} $\idenot{\Gamma}{r}\in\idenot{\Gamma}{\phi}$ inductively by:
\begin{frameqn}
\idenot{\Gamma,a{:}\phi}{a} = a^\iden_\phi 
\qquad
\idenot{\Gamma}{C} = C^\iden
\qquad
\idenot{\Gamma}{rs} = \idenot{\Gamma}{r}\bullet\idenot{\Gamma}{s}
\qquad
\idenot{\Gamma}{\lam{a{:}\phi}r} = [a{:}\phi]\idenot{\Gamma,a{:}\phi}{r} 
\end{frameqn}
\end{defn}

We now come to our first soundness theorem; if a term is typable then its denotation inhabits the denotation of its type:
\begin{thrm}[\bf First soundness theorem]
\label{thrm.type.soundness}
If $\Gamma\cent r:\phi$ then $\idenot{\Gamma}{r}\in\idenot{\Gamma}{\phi}$.

As a corollary using condition~5 of Definition~\ref{defn.lambda.model}, if $\Gamma\cent r:\phi$ then $\supp(\idenot{\Gamma}{r})\subseteq\fa(r)$.
\end{thrm}
\begin{proof}
By a straightforward induction on the derivation of $\Gamma\cent r:\phi$:
\begin{itemize*}
\item
By the definition of model (Definition~\ref{defn.lambda.model}), if $a:\phi\in\Gamma$ then $\idenot{\Gamma}{a}=a^\iden_\phi\in\idenot{\Gamma}{\phi}$, and $\idenot{\Gamma}{C}=C^\iden\in\idenot{\Gamma}{\f{type}(C)}$.
\item
Suppose $\Gamma,a{:}\phi\cent r:\psi$ so that by \rulefont{L} $\Gamma\cent \lam{a{:}\phi}r:\phi\fto\psi$ and by inductive hypothesis $\idenot{\Gamma,a{:}\phi}{r}\in\idenot{\Gamma}{\psi}$.
By assumpion $\idenot{\Gamma}{\lam{a{:}\phi}r}\in\idenot{\Gamma}{\phi\fto\psi}$.
\item
If $\idenot{\Gamma}{r}\in\idenot{\Gamma}{\phi\fto\psi}$ and $\idenot{\Gamma}{s}\in\idenot{\Gamma}{\phi}$ then $\idenot{\Gamma}{r}\bullet\idenot{\Gamma}{s}\in\idenot{\Gamma}{\phi}$.
\qedhere\end{itemize*}
\end{proof}

\subsection{Soundness for $\beta$-conversion}

\begin{lemm}
\label{lemm.st.beta.1}
Suppose $\Gamma,a{:}\phi\cent r:\psi$ and $\Gamma\cent s:\phi$, where ${a\not\in\dom(\Gamma)}$.

Then $\idenot{\Gamma}{r[a\ssm s]}=\idenot{\Gamma,a{:}\phi}{r}[a\sm\idenot{\Gamma}{s}]$.
\end{lemm}
\begin{proof}
By a routine induction on the derivation of $\Gamma,a{:}\phi\cent r:\psi$:
\begin{itemize*}
\item \emph{The case of \rulefont{V} for $a$.}\quad
By \rulefont{Suba} $a^\iden_\phi[a\sm\idenot{\Gamma}{s}]=\idenot{\Gamma}{s}$.
Also $a[a\ssm s]=s$.
\item \emph{The case of \rulefont{V} for $c{:}\chi\in\Gamma$.}\quad
By \rulefont{Sub\#} $c^\iden_\chi[a\sm\idenot{\Gamma}{s}]=c^\iden_\chi$.
Also $c[a\ssm s]=s$.
\item \emph{The case of \rulefont{L} for $\lam{c{:}\chi}r$.}\quad
By equivariance (Theorem~\ref{thrm.equivar}) suppose $c\not\in\dom(\Gamma)\cup\supp(\idenot{\Gamma}{s})$.\footnote{In fact by Theorem~\ref{thrm.type.soundness} $c\not\in\supp(\idenot{\Gamma}{s})$ follows from $c\not\in\dom(\Gamma)$.  But that does not matter; we can just rename $c$ `fresh', without having to engage in detailed calculations about how fresh it is.}
The result follows using \rulefont{Sub\text{$\lambda$}}.
\item
\emph{The case of \rulefont{A}} uses \rulefont{SubApp}.
\qedhere\end{itemize*}
\end{proof}

\begin{prop}[\bf The $\xi$ rule]
\label{prop.xi}
Suppose $\Gamma,a{:}\phi\cent r:\psi$ and $\Gamma,a{:}\phi\cent s:\psi$.

If $\idenot{\Gamma,a{:}\phi}{r}=\idenot{\Gamma,a{:}\phi}{s}$ then $\idenot{\Gamma}{\lam{a{:}\phi}r}=\idenot{\Gamma}{\lam{a{:}\phi}{s}}$.
\end{prop}
\begin{proof}
Immediate since by Definition~\ref{defn.interp} $\idenot{\Gamma}{\lam{a{:}\phi}r}=[a{:}\phi]\idenot{\Gamma,a{:}\phi}{r}$, and similarly for $s$.
\end{proof}

\begin{xmpl}
\label{xmpl.two.point}
Proposition~\ref{prop.xi} does not hold in a valuation semantics of `ordinary' models.
For instance, consider a (valuation-based) semantics with one base type $\tau$ with denotation $\{0,1\}$ (a two-element set).
Consider $x$ and $y$ and a valuation $\rho$ mapping $x$ and $y$ both to $0$.
Then $\denot{}{\rho}{x}=0=\denot{}{\rho}{y}$ but $\denot{}{\rho}{\lam{x{:}\tau}x}\neq\denot{}{\rho}{\lam{x{:}\tau}y}$.
\end{xmpl}


\begin{corr}[\bf Second soundness theorem]
If $r\beq s$ (Defn.~\ref{defn.beta.equivalence}) and $\Gamma\cent r:\phi$ then $\idenot{\Gamma}{r}=\idenot{\Gamma}{s}$.
\end{corr}
\begin{proof}
By some routine sets calculations, using Lemma~\ref{lemm.st.beta.1} and Proposition~\ref{prop.xi}.
\end{proof}

\subsection{Completeness}
\label{subsect.completeness}

\begin{defn}
\label{defn.ment}
Write $\interp I;\Gamma\ment r{=}s$ when there exists $\phi$ (which is unique if it exists) such that $\Gamma\cent r{:}\phi$ and $\Gamma\cent s{:}\phi$,\ and $\idenot{\Gamma}{r}{=}\idenot{\Gamma}{s}$ and $\idenot{\Gamma}{r}\in\idenot{\Gamma}{\phi}$. 
We call the \deffont{(typed) equality} $\Gamma\cent r{=}s$ \deffont{valid} in $\interp I$.
\end{defn}


We need one technical fact about nominal sets, for Theorem~\ref{thrm.st.complete}: 
\begin{lemm}
\label{lemm.beeq.fresh.rep}
Suppose $\Gamma\cent r:\phi$.
If $a\not\in\supp(\beeq{r})$ then there exists $s$ such that $\theory T;\Gamma\cent r=s$ and $a\not\in\fa(s)$.
\end{lemm}
\begin{proof}
Using \cite[Lemma~7.6.2]{gabbay:nomtnl}.
\end{proof}

\begin{thrm}
\label{thrm.st.complete}
$\Gamma\ment r=s$ implies $r\beq s$.
\end{thrm}
\begin{proof}
We take as our model $\interp I$ where $\idenot{\Gamma}{r}=\beeq{r}$ and $\idenot{\Gamma}{\phi}=\{\beeq{r}\mid \Gamma\cent r:\phi\}$, and:
\begin{itemize*}
\item
If $a:\phi\in\Gamma$ then we take $a^\iden_\phi=\beeq{a}\in\idenot{\Gamma}{\phi}$.
\item
If $\beeq{r}\in\idenot{\Gamma,a{:}\phi}{\psi}$ then we take $[a{:}\phi]\beeq{r}=\beeq{\lam{a{:}\phi}r}$.
\item
Similarly, we take $\beeq{r}\bullet \beeq{s}=\beeq{rs}$.
\end{itemize*}
It is a fact that $r\beq r'$ and $s\beq s'$ imply $\lam{a{:}\phi}r\beq \lam{a{:}\phi}r'$ and $rs\beq r's'$, and it follows that the definition above is well-defined.

We must also check validity of rules \rulefont{Suba} to \rulefont{Sub\text{$\lambda$}}. 
We consider two cases:
\begin{itemize*}
\item
\emph{The case of \rulefont{Suba}.}\quad
It is a fact that $(\lam{a{:}\phi}a)s\beq s$. 
\item
\emph{The case of \rulefont{Sub\#}.}\quad
Suppose $a\#\beeq{t}$.
By Lemma~\ref{lemm.beeq.fresh.rep} there exists $t'\beq t$ such that $a\not\in\fa(t')$.
So $t'[a\ssm r]=t'$ and thus $\theory T;\Gamma\cent t'[a\ssm r]= t'$.
It follows 
that $\beeq{t[a\ssm r]}=\beeq{t}$. 
\end{itemize*}

Furthermore, by construction if $\beeq{r}=\beeq{s}$ then $r\beq s$. 
\end{proof}
The proof of Theorem~\ref{thrm.st.complete} resembles the proof of completeness for Henkin models, with moderate changes to handle the `nominal' models.
Our models are not necessarily extensional (that is, we do not insist that $r=\lam{a}(ra)$ for $a$ not free in $r$) whereas Henkin semantics usually are \cite{henkin:comtot}; nevertheless it is reasonable to think of this as `Henkin semantics with names'.
A survey of complete non-extensional semantics for STLC is in \cite{benzmuller:higose}.

Theorem~\ref{thrm.st.complete} is simpler than it could be; we could generalise it to completeness for arbitrary theories (i.e. we allow a set of equality axioms and prove completeness for the class of models that validate those axioms).  
We expect this generalisation to be an easy replay of the existing proof.
We do not do this because the simpler case already illustrates the main points, and has useful features which we can now explore.

\subsection{Well-pointedness}
\label{subsect.well-pointed}

In Proposition~\ref{prop.xi} and Example~\ref{xmpl.two.point} we saw that our nominal Henkin models have a desirable property that `ordinary' models do not.
We now come to another; to state it we need a definition:
\begin{defn}
A \deffont{homomorphism} $F$ from $\interp I$ to $\interp J$ is a collection of functions $F^\phi_{\Gamma}$ mapping $\idenot{\Gamma}{\phi}$ to $\jdenot{\Gamma}{\phi}$ which are:
\begin{itemize*}
\item
\emph{Equivariant} in the sense that $\pi\act (F^\phi_{\Gamma}(x))=F^\phi_{\pi\act\Gamma}(\pi\act x)$ (Notation~\ref{nttn.pi.Gamma}).
\item
\emph{Natural} in the sense that $F$ commutes with atoms, constants, abstraction, and $\bullet$.
So for example, $F^\phi_{\Gamma,a{:}\phi}(a^\iden_\phi)=a^\jden_\phi$ and $F^{\psi\fto\phi}_{\Gamma}([a{:}\psi]x)=[a{:}\psi]F^\phi_{\Gamma,a{:}\phi}(x)$. 
\end{itemize*}
\end{defn}

The notion of validity from Definition~\ref{defn.ment} is \emph{local} in that it checks validity at one model.
There is also a \emph{global} notion, which checks validity at the model and all `larger' ones:
\begin{defn}
Suppose $\Gamma\cent r{:}\phi$ and $\Gamma\cent s{:}\phi$.
Say that $\interp I;\Gamma\gloment r=s$ when $\interp J;\Gamma\ment r=s$ in the sense of Definition~\ref{defn.ment} for every $\interp J$ such that there exists a homomorphism $F:\interp I\to\interp J$. 
\end{defn}

\begin{lemm}
\label{lemm.sub.morph.commute}
Suppose $\Gamma\cent r:\phi$ and $F:\interp I\to\interp J$ is a homomorphism.
Then 
$F^\phi_{\Gamma}(\idenot{\Gamma}{r})=\jdenot{\Gamma}{r}$.
\end{lemm}
\begin{proof}
By a routine induction on the derivation of $\Gamma\cent r:\phi$, using naturality.
\end{proof}

\begin{thrm}[\bf Well-pointedness]
\label{thrm.well-pointed}
Suppose $\Gamma\cent r:\phi$ and $\Gamma\cent s:\phi$. 
Then $\interp I;\Gamma\ment r=s$ if and only if $\interp I;\Gamma\gloment r=s$. 
\end{thrm}
\begin{proof}
By considering the \emph{identity} homomorphism from $\interp I$ to itself, mapping $x\in\idenot{\Gamma}{\phi}$ to itself, it is clear that $\interp I;\Gamma\gloment r=s$ implies $\interp I;\Gamma\ment r=s$.

Conversely, suppose $\interp I;\Gamma\ment r=s$ and suppose $F$ is a homomorphism from $\interp I$ to $\interp J$.
By assumption $\idenot{\Gamma}{r}=\idenot{\Gamma}{s}$.
It follows by Lemma~\ref{lemm.sub.morph.commute} that $\jdenot{\Gamma}{r}=F^\phi_\Gamma(\idenot{\Gamma}{r})=F^\phi_\Gamma(\jdenot{\Gamma}{s})=\jdenot{\Gamma}{s}$.
\end{proof}

\begin{xmpl}
\label{xmpl.one.point}
Theorem~\ref{thrm.well-pointed} fails for traditional models.
For instance, consider a functional model in which all terms are equal because every type has just one element.
So $\Gamma\cent r=s$ is locally true, but not globally true.

Nominal Henkin semantics exclude this, because they have elements to interpret variables.
It is impossible to compress them all down to one element, as we did in the previous paragraph for `ordinary' models. 
\end{xmpl}

\section{Existential variables}
\label{sect.existential.variables}

Nominal terms introduced to nominal techniques the idea of two levels of variable; atoms (as above) and \emph{unknowns} $X$, which are existential variables and in \cite{gabbay:nomu-jv} were used in a unification algorithm. 
The first author proposed combining nominal unknowns with non-trivial logical theories, e.g. with first-order logic \cite{gabbay:oneaah,gabbay:oneaah-jv}.
Since in this paper we have a nominal semantics for the STLC, it is natural to extend Definition~\ref{defn.pts.sorts} with nominal unknowns and so to add existential variables.

The motivation for doing this is that STLC underlies many interesting logics and programming languages, so that our semantics and syntax with existential variables have potential---not exploited in this paper but motivating the constructions---to provide syntax and semantics for `incomplete terms'.
In common with all other such treatments, a difficulty is the delicacy of maintaining well-typedness under instantiation (which for nominal terms may be capturing; see Remark~\ref{rmrk.why.down}).
Our solution has elements of previous work, but it retains a distinct identity and remains typically `nominal'.

We will use \emph{permissive} nominal terms~\cite{gabbay:perntu-jv}, which improve on the theory of $\alpha$-equivalence of nominal terms by allowing us to `just quotient' syntax (nominal terms require a freshness context and freshness context update, which are harder to manage in the presence of non-trivial equalities/reductions on terms). 
\subsection{Syntax}
\label{subsect.syntax} 

\begin{defn}
\label{defn.partition}
Fix a partition of the set of atoms from Definition~\ref{defn.atoms} into two disjoint countably infinite sets $\atomsdown$ and $\atomsup$, so that $\mathbb A=\atomsdown\uplus\atomsup$ 
\end{defn}

Splitting $\mathbb A$ in two is key to the syntax, but not to the semantics: the notion of model in Definition~\ref{defn.holes.lambda.model} is identical to Definition~\ref{defn.lambda.model} and is based on finitely-supported nominal sets as usual.
Only the syntax uses the more powerful notions of $\atomsdown$ and $\atomsup$ (and it is more powerful; see e.g. Example~\ref{xmpl.hole.power}).
This echoes the formal distinction between `names that exist to be bound' and `names that exist to be free' used in some treatments of logic \cite{heijenoort:fregsb,smullyan:firol}, though this distinction is not so rigid here; e.g. a `standalone atom' $a$ can appear either from $\atomsdown$ or $\atomsup$ and via a permutation or substitution `migrate' between them.


\begin{defn}
Fix a countably infinite set of \deffont{unknowns}.
$X$, $Y$, $Z$ will range over distinct unknowns.
\end{defn}

\begin{defn}
\label{defn.holes.sorts}
Types are as in Definition~\ref{defn.pts.sorts}.
Terms are defined by:
\begin{displaymath}
\begin{array}{r@{\ }l}
\rtm      ::= & a \mid C \mid X[b_i{\ssm} s_i]_{i{=}1}^n \mid \lam{c{:}\phi}\stm \mid \rtm\stm 
\qquad (\{b_i\mid 1{\leq}i{\leq}n\}\subseteq\atomsdown,\ \ c\in\atomsup)
\end{array}
\end{displaymath}
$[b_i\ssm s_i]$ is a \deffont{(level 1) substitution}, which is a finite partial function from atoms to terms, mapping $b_i$ to $s_i$ and undefined elsewhere (so finite substitutions are directly in this syntax, just like finite permutations on unknowns $\pi\act X$ are directly part of nominal terms).
We call $X[b_i{\ssm}s_i]_1^n$ a \deffont{moderated unknown}.
\end{defn}
The condition $c\in\atomsup$ may seem odd---so $\lam{a{:}\phi}a$ is not well-formed syntax if $a\in\atomsdown$---but since $a$ is supposed to be bound we can intuitively always $\alpha$-convert it to be in $\atomsup$.  This is a useful `hygiene' simplification, since just by looking at an atom $a$ we can tell if it could be bound ($a\in\atomsup$) or captured by an instantiation ($a\in\atomsdown$).
We can always move between one world and the other using a moderating substitution, as in $\lam{a{:}\phi}X[b{\ssm} a]$ where $a\in\atomsup$ and $b\in\atomsdown$.

\begin{xmpl}
\begin{itemize*}
\item
\emph{An incomplete term.}\quad
The typing $a,b{:}\phi,X{:}\phi\cent \lam{a{:}\phi}X[b{\ssm}a]:\phi\fto\phi$ where $a\in\atomsup$ and $b\in\atomsdown$ represents an incomplete typing `$\lam{x{:}\phi}t$ where $t$ has type $\phi$'. 
This is an term for a function on one argument.

Looking forward to the level 1 and 2 substitutions in Definitions~\ref{defn.holes.sub} and~\ref{defn.st.sub.2}, we will be able to complete $\lam{a{:}\phi}X$ to a complete term, by applying the substitution $[X\ssm b]$.
We get the identity $\lam{a{:}\phi}a$.
Without unknowns, both the incomplete and the complete terms would be represented by $\lam{a{:}\phi}fa$ for a higher-order $f:\phi\fto\phi$ (whereas $X$ has type $\phi$).
\item
\emph{An incomplete HOL predicate.}\quad
Assume base types $\iota$ and $o$ and constants $\dot\limp:o\fto o\fto o$ and $\dot\forall:(\iota\fto o)\fto o$.
The typing $X:o,\ Y:\iota,\ b{:}\iota\cent (\dot\forall\lam{b{:}\iota}X)\dot\limp X[b\ssm Y]:o$ represents an incomplete HOL predicate. 

Without level 2 variables, both the incomplete and the complete terms would be represented by $b{:}\iota,\ f:\iota\fto o\cent (\dot\forall\lam{b{:}\iota}fa)\dot\limp fa$.
\end{itemize*} 
\end{xmpl} 

\begin{defn}
\label{defn.perm.action}
Suppose a permutation $\pi$ (Definition~\ref{defn.swap}) is such that $\f{nontriv}(\pi)\subseteq\atomsup$. 
Define a \deffont{permutation action} $\pi\act r$ on terms by:
$$
\begin{array}{r@{\ }l@{\qquad}r@{\ }l@{\qquad}r@{\ }l}
\pi\act a=&\pi(a) 
&
\pi\act\lam{a{:}\phi}r=&\lam{\pi(a){:}\phi}\pi\act r
&
\pi\act C=&C
\\
\pi\act (rs)=&(\pi\act r)(\pi\act s)
&
\pi\act (X\,[b_i\ssm s_i]_i) =&X\,[b_i\ssm \pi\act s_i]_i
\end{array}
$$
\end{defn}

\begin{rmrk}
\label{rmrk.simple}
Intuitively, the reason that we restrict $\nontriv(\pi)$ to atoms in $\atomsup$ is so that we only rename the atoms that can be $\lambda$-abstracted.
This restriction could be removed, and the syntax made `more equivariant', but at the price of complicating the syntax $X[b_i\ssm s_i]_1^n$ to $(\pi\act X)[b_i\ssm s_i]_1^n$ so that we could write $\pi\act (X\,[b_i\ssm s_i]_i = (\pi\act X)[\pi(b_i){\ssm}\pi\act s_i]$.
There would be nothing wrong with this---it just makes our basic syntax slightly more complicated.
Since there is no change in expressivity, we leave this out.

We could also emulate $\pi$ using the substitution $[b_i\ssm s_i]_1^n$, but then we must add \rulefont{SubId}.
\end{rmrk}

\begin{defn}
\label{defn.holes.congruence}
Call a binary relation $\somerel$ on terms a \deffont{congruence} when it is closed under the rules \rulefont{CongApp} and \rulefont{\xi} in Definition~\ref{defn.congruence} and in addition:\footnote{The condition $\{b_1,\dots,b_n\}\subseteq\atomsdown$ is there to guarantee that $X[b_i{\ssm}s_i]$ and $X[b_i{\ssm}s_i']$ are well-formed terms.}
$$
\begin{prooftree}
s_i\somerel s_i'\quad (1{\leq}i{\leq}n,\ \{b_1,\dots,b_n\}\subseteq\atomsdown)
\justifies
X[b_i{\ssm}s_i]_1^n \somerel X[b_i{\ssm}s_i']_1^n
\using\rulefont{CongX}
\end{prooftree}
$$
\end{defn}

\begin{defn}
$\alpha$-equivalence $r\aeq s$ is the least congruence such that if $b\in\atomsup\setminus\fa(r)$ then $\lam{a{:}\phi}r\aeq\lam{b{:}\phi}(b\ a)\act r$.
Henceforth we quotient terms by $\alpha$-equivalence.
\end{defn}

\begin{defn}
\label{defn.hole.free.atoms}
Define \deffont{free atoms} $\fa(\rtm)$ and \deffont{free unknowns} $\fv(\rtm)$ by:
\begin{displaymath}
\begin{array}{r@{\ }l@{\qquad}r@{\ }l@{\quad}r@{\ }l}
\fa(a) = & \{ a \}                               
& 
\fa(\lam{a{:}\phi}r) = & \fa(r) \setminus \{ a \}
&
\fa(C) = & \varnothing
\\
\fa(\rtm\stm) = & \fa(\rtm) \cup \fa(\stm) 
&
\fa(X[b_i{\ssm}s_i]_1^n) = & (\atomsdown{\setminus}\{b_i\mid i\})\cup\bigcup_i\fa(s_i)
\\[2ex]
\fv(a) = & \varnothing
& 
\fv(\lam{a{:}\rtm}\stm) = & \fv(\stm) \cup \fv(\rtm) 
&
\fv(C) = & \varnothing
\\
\fv(\rtm\stm) = & \fv(\rtm) \cup \fv(\stm)  
&
\fv(X[a_i{\ssm} s_i]_1^n) = & \{ X \}\cup\bigcup_i\fv(s_i) 
\end{array}
\end{displaymath}
\end{defn}

\begin{defn}
\label{defn.holes.sub}
We give terms a capture-avoiding \deffont{substitution action} $r[b_i\ssm s_i]_1^n$ as follows:
$$
\begin{array}{r@{\ }l@{\qquad}l}
b_j[b_i{\ssm} s_i]_1^n=&s_j & (1{\leq}j{\leq}n)
\\
a[b_i{\ssm} s_i]_1^n=&a & (a\not\in\{b_i\mid 1{\leq}i{\leq}n\})
\\
C[b_i{\ssm} s_i]_1^n=&C
\\
X[b_i{\ssm} s_i]_{i\in A}[b_i{\ssm} s_i]_{i\in B}=&
X[\begin{array}[t]{l} \!\!\! (b_i{\ssm}s_i)_{i\in B{\setminus}A,\ b_i{\in}\atomsdown}, \\ (b_i{\ssm}s_i[b_j{\ssm}s_j]_{j\in B})_{i\in A}] \end{array}
&(B=\{1,\dots,n\})
\\
(\lam{c{:}\phi}r)[b_i{\ssm}s_i]_1^n=&\lam{c{:}\phi}(r[b_i{\ssm}s_i]_1^n) &(c\in\atomsup\setminus(\bigcup_i\{b_i\}\cup\fa(s_i)))
\\
(r'r)[b_i{\ssm} s_i]_1^n=&(r'[b_i{\ssm}s_i]_1^n)(r[b_i{\ssm}s_i]_1^n)
\end{array}
$$ 
\end{defn}
\noindent Note above that if $b_i\in\atomsup$ then it gets garbage-collected (eliminated) on $X$, as we see from the condition `$b_i\in\atomsdown$' in `$i\in B{\setminus}A,\ b_i\in\atomsdown$'. 
So for instance $X[b{\ssm}b'][b'{\ssm}b'']=X[b{\ssm}b'',b'{\ssm}b'']$ where $b,b',b''\in\atomsdown$ and $X[b{\ssm}a][a{\ssm}b'']=X[b{\ssm}b'']$ where $a\in\atomsup$.

\begin{defn}
\label{defn.holes.beta.equivalence}
Let \deffont{$\beta$-equivalence} $- \beq -$ be the least congruence (Definition~\ref{defn.holes.congruence}) such that
$(\lam{a{:}\phi}\rtm) \ttm\beq \rtm[a \ssm \ttm]$.
\end{defn}

\subsection{Environments and typing}
\label{subsect.with.holes.typing}

\begin{defn}
\label{defn.holes.environments}
A \deffont{type environment} $\Gamma$ is a set of \deffont{atomic typings} $a:\phi$ or $X:\phi$ which is \emph{functional} in the sense that if $a:\phi$ and $a:\phi'$ then $\phi=\phi'$, and similarly for $X$ (i.e. `add $X{:}\phi$ to Definition~\ref{defn.st.environments}').

Define $\dom(\Gamma)=\{a\mid \Exists{\phi}a{:}\phi\in\Gamma\}\cup\{X\mid\Exists{\phi}X{:}\phi\in\Gamma\}$.
\end{defn}


\begin{figure}
\begin{displaymath}
\begin{array}{c}
\begin{prooftree}
(a:\phi \in \Gamma)
\justifies
\Gamma \cent a : \phi
\using\rulefont{V}
\end{prooftree}
\qquad
\begin{prooftree}
(\f{type}(C)=\phi)
\justifies
\Gamma \cent C:\phi
\using\rulefont{C}
\end{prooftree}
\qquad
\begin{prooftree}
\Gamma,a{:}\phi\cent \rtm:\psi\quad (a\in\atomsup) 
\justifies
\Gamma\cent (\lam{a{:}\phi}\rtm) : \phi\fto\psi 
\using\rulefont{L}
\end{prooftree}
\\[5ex]
\begin{prooftree}
\Gamma \cent \rtm : \phi\fto\psi 
\quad 
\Gamma \cent \stm : \phi 
\justifies
\Gamma \cent \rtm\stm : \psi
\using\rulefont{A}
\end{prooftree}
\qquad
\begin{prooftree}
\Gamma \cent s_i : \psi_i 
\quad (X{:}\phi{\in}\Gamma,\ b_i{:}\psi_i{\in}\Gamma,\ 1{\leq}i{\leq}n)
\justifies
\Gamma \cent X\,[b_i {\ssm} s_i]_1^n : \phi
\using\rulefont{Meta}
\end{prooftree}
\end{array}
\end{displaymath}
\caption{Typing rules for the simply-typed $\lambda$-calculus with holes}
\label{fig.hole.typing.rules}
\end{figure}

\begin{defn}
\label{defn.hole.typing.rules}
Define a \deffont{typing relation} by the rules in Figure~\ref{fig.hole.typing.rules}.
\end{defn}

One interesting feature of Figure~\ref{fig.hole.typing.rules} is that $b_i$ \emph{must} be typed in $\Gamma$ in \rulefont{Meta}.
This means that we can strengthen only for atoms in $\atomsup$ (the `abstractable' atoms); see Lemma~\ref{lemm.holes.strong}.
See Remarks~\ref{rmrk.why.down} and~\ref{rmrk.discussion.why} for discussions of why.
Also, the $s_i$ are typed in a context in which the $b_i$ occur.
There is no problem with circularities; in the models the $b_i$ are just elements (with special properties).
 
\begin{lemm}[Weakening]
\label{lemm.holes.weaken}
If $\Gamma\cent r:\psi$ and $c\not\in\dom(\Gamma)$ then $\Gamma,c{:}\chi\cent r:\psi$.
\end{lemm}
\begin{proof}
By induction on the derivation of $\Gamma\cent r:\psi$.
For the case of \rulefont{L} we may rename using equivariance (Theorem~\ref{thrm.equivar}).
\end{proof} 

\begin{lemm}[Strengthening]
\label{lemm.holes.strong}
If $\Gamma,c{:}\chi\cent r:\psi$ and $c\in\atomsup\setminus\fa(r)$ then $\Gamma\cent r:\psi$.
\end{lemm}
\begin{proof}
By induction on the derivation of $\Gamma\cent r:\psi$.
For the case of \rulefont{L} we may rename using equivariance (Theorem~\ref{thrm.equivar}).
The rule \rulefont{Meta} is why we insist on $c{\not\in}\atomsdown$; the atoms $\{b_i\mid 1{\leq}i{\leq}n\}\subseteq\atomsdown$ may not feature in $\fa(r)$ but \emph{must} be in $\Gamma$, as discussed in Remark~\ref{rmrk.why.down}.
\end{proof} 

\begin{rmrk}
\label{rmrk.why.down}
\rulefont{Meta} states that if $X{:}\phi\in\Gamma$ and $b_i{:}\psi_i\in\Gamma$ for $1{\leq}i{\leq}n$ then $\Gamma\cent s_i:\psi_i$ for $1{\leq}i{\leq}n$ implies $\Gamma\cent X[b_i{\ssm}s_i]:\phi$.
We must insist on $b_i{:}\psi_i\in\Gamma$, for suppose otherwise: 
Then for $a\in\atomsup$ and $b\in\atomsdown$ we could derive $X{:}\phi,a{:}\chi\cent \lam{b{:}\phi}X[a\ssm b]:\phi\fto\phi$; the types of $a$ and $b$ are inconsistent.
\end{rmrk}

\begin{lemm}
\label{lemm.level.one.subst.reduces.free.atoms} 
\label{lemm.holes.fa.sub}
$\f{fa}(\rtm[a \ssm \ttm]) \subseteq (\f{fa}(\rtm) \setminus \{ a \}) \cup \f{fa}(\ttm)$
\end{lemm}

\begin{lemm}
\label{lemm.hole.sub.single}
$\Gamma,(c_i{:}\chi_i)_1^n\cent r:\phi$ and $\Gamma,(c_i{:}\chi_i)_1^n\cent \stm_j:\chi_j$ for $1{\leq}j{\leq}n$ imply $\Gamma,(c_i{:}\chi_i)_1^n\cent \rtm[c_i\ssm\stm_i]_1^n:\phi$.

As a corollary, if $c\in\atomsup\setminus\fa(s)$ then $\Gamma,c{:}\chi\cent r:\phi$ and $\Gamma\cent s:\chi$ imply $\Gamma\cent r[c\ssm s]:\phi$. 
\end{lemm}
\begin{proof}
By a routine induction on the derivation of $\Gamma,(c_i{:}\chi_i)_1^n\cent r:\phi$.
For the case of \rulefont{L} we may rename the bound atom in the derivation using equivariance (Theorem~\ref{thrm.equivar}).

The corollary follows by Lemmas~\ref{lemm.holes.weaken}, \ref{lemm.holes.strong}, and~\ref{lemm.holes.fa.sub}.
\end{proof}

\begin{rmrk}
\label{rmrk.discussion.why}
Lemma~\ref{lemm.hole.sub.single} does not state $\Gamma,c{:}\chi\cent \rtm:\phi$ and 
$\Gamma\cent \stm:\chi$ 
imply 
$\Gamma\cent \rtm[c\ssm\stm]:\phi$, for $c\in\atomsdown$.
For instance $X{:}\phi,c{:}\phi\cent X:\phi$ and $X{:}\phi\cent X:\phi$ but it is not the case that $X{:}\phi\cent X[c\ssm X]:\phi$.
\end{rmrk}

\section{Level 2 substitution}
\label{sect.level.2.substitution}

\begin{defn}
\label{defn.st.sub.2}
A \deffont{level 2 substitution} is a map $\theta$ from unknowns to terms.\footnote{The reader familiar with nominal techniques might expect a condition that $\fa(\theta(X))\subseteq\atomsdown$ always.
This would be necessary if moderations were permutations, but is not if they are substitutions.
See \cite[Proposition~3.4.3]{gabbay:nomtnl}.}
Write $[X\ssm t]$ for the substitution mapping $X$ to $t$ and all other $Y$ to $Y$.
$$
\begin{array}{r@{\ }l@{\qquad}r@{\ }l@{\qquad}r@{\ }l}
a\theta          = & a             
& 
(\lam{a{:}\phi}\stm)\theta = & \lam{a{:}\phi}(\stm\theta) \quad (a\in\atomsup{\setminus}\bigcup_{X\in\fv(t)}\fa(\theta(X)))
&
C\theta =&C
\\
(\rtm\stm)\theta = & (\rtm\theta)(\stm\theta) 
& 
X[b_i{\ssm}s_i]_1^n\theta = & \theta(X)[b_i{\ssm}s_i\theta]_1^n
\end{array}
$$
\end{defn}

\begin{prop}
\label{prop.st.sub2}
If $\Gamma,X{:}\xi\cent \rtm{:}\phi$
and
$\Gamma\cent \ttm{:}\xi$
then
$\Gamma\cent \rtm[X\ssm\ttm]{:}\phi$.
\end{prop}
\begin{proof}
By induction on the derivation of $\Gamma,X{:}\xi\cent\rtm{:}\phi$.  The cases of \rulefont{V} and \rulefont{A} are routine: 
\begin{itemize*}
\item
\emph{The case of \rulefont{Meta} for $X$.}\quad
Suppose $\Gamma,X{:}\xi,(b_i{:}\psi_i)_1^n\cent X[b_i{\ssm} s_i]_1^n:\xi$ because $\Gamma,X{:}\xi,(b_i{:}\psi_i)_1^n\cent s_j:\psi_j$ for $1{\leq}j{\leq}n$.

Suppose $\Gamma,(b_i{:}\psi_i)_1^n\cent t:\xi$.
By inductive hypothesis $\Gamma,(b_i{:}\psi_i)_1^n\cent s_j[X{\ssm}\ttm]:\psi_j$ for $1{\leq}j{\leq}n$.
By definition $X[b_i{\ssm}s_i]_1^n[X{\ssm}\ttm]=\ttm[b_i{\ssm} s_i[X{\ssm}\ttm]]_1^n$.
So it suffices to show that 
$\Gamma,(b_i{:}\psi_i)_1^n\cent t[b_i{\ssm} s_i[X{\ssm}t]]_1^n:\xi$.
We use Lemma~\ref{lemm.hole.sub.single}.
\item
\emph{The case of \rulefont{L}.}\quad
Renaming if necessary using equivariance (Theorem~\ref{thrm.equivar}), assume $a\in\atomsup\setminus\fa(t)$.
By definition $(\lam{a{:}\phi}\rtm)[X{\ssm} \ttm]=\lam{a{:}\phi}(\rtm[X{\ssm}\ttm])$.
We use the inductive hypothesis for $\Gamma,X{:}\xi,a{:}\phi\cent \rtm:\psi$. 
%
%
%
%
%
\qedhere \end{itemize*}
\end{proof}

\subsection{Models}

\begin{defn}
\label{defn.holes.lambda.model}
A \deffont{model} $\interp I$ consists of an assignment for each type environment $\Gamma$ and type $\phi$ of a finitely-supported set $\idenot{\Gamma}{\phi}$ together with the same data as in Definition~\ref{defn.lambda.model}, satisfying the same equivariance conditions and \rulefont{Suba} to \rulefont{Sub{\text{$\lambda$}}} except that in addition:
\begin{enumerate*}
\setcounter{enumi}{7}
\item
If $\Gamma(a)=\Gamma'(a)$ for every $a\in\mathbb A$ then $\idenot{\Gamma}{\phi}=\idenot{\Gamma'}{\phi}$ (so the model ignores $X{:}\phi\in\Gamma$ and only looks at the typing of atoms).
\end{enumerate*}
\end{defn}

\begin{defn}[\bf Simultaneous substitution]
\label{defn.sim.sub}
Suppose $b_i{:}\psi_i\in\Gamma$ and $y_i\in\idenot{\Gamma}{\psi_i}$ for $1{\leq}i{\leq}n$.
Suppose $x\in\idenot{\Gamma}{\phi}$.
Specify $x[b_i{\sm}y_i]_1^n$ to be equal to 
$\bigl(((c_1\ b_1)\circ\dots\circ(c_n\ b_n))\act x\bigr)[c_1{\sm}y_1]\dots[c_n{\sm}y_n]$
for fresh $c_1,\dots,c_n$ (so $c_i\not\in\supp(x)\cup\bigcup_i\supp(y_i)$ for $1\leq i\leq n$).\footnote{Definition~\ref{defn.holes.lambda.model} only provides substitution for one atom at a time.  
We need simultaneous substitution in the semantics to give meaning to level 2 variables (see Definition~\ref{defn.hole.interp}).
The minor difficulty is that it might be that $b_i\in\supp(y_j)$.  
So we `rename atoms fresh' first, and then substitute for these atoms one at a time.
Certain detailed but routine verifications are necessary to make sure this works and is well-defined (depends neither on the fresh choice of $c_i$, nor on the order in which the substitutions are then carried out).
The relevant maths is described in \cite[Section~6]{gabbay:fountl}.}
\end{defn}

\begin{lemm}
\label{lemm.intype.equivar}
If $x\in\idenot{\Gamma}{\phi}$ then $\pi\act x\in\idenot{\pi\act\Gamma}{\phi}$.
\end{lemm}
\begin{proof}
Direct from equivariance (Theorem~\ref{thrm.equivar}).
\end{proof}

By Lemma~\ref{lemm.intype.equivar} syntax is equivariant for atoms in $\atomsup$, because the predicate we use to define it in Definition~\ref{defn.holes.sorts} uses a partition $\mathbb A=\atomsdown\cup\atomsup$.  
The notion of model of Definition~\ref{defn.holes.lambda.model} does not use this partition however, so it is equivariant for \emph{all} $\pi$, and not just those with $\nontriv(\pi)\subseteq\atomsup$.  
Lemma~\ref{lemm.intype.equivar} depends on this, and Lemma~\ref{lemm.holes.easy} uses it.

\begin{lemm}
\label{lemm.holes.easy}
Suppose $b_i{:}\psi\in\Gamma$ and $y_i\in\idenot{\Gamma}{\psi}$ for $1{\leq}i{\leq}n$, and suppose $z\in\idenot{\Gamma}{\chi}$.
Then $z[b_i{\sm} y_i]_1^n\in\idenot{\Gamma}{\chi}$.
\end{lemm}
\begin{proof}
Unpack Definition~\ref{defn.sim.sub} and use Lemma~\ref{lemm.intype.equivar} and conditions~2 and~3 of Definition~\ref{defn.holes.lambda.model}.
\end{proof}

%
%
%

\begin{defn}
\label{defn.holes.valuation}
A \deffont{valuation} $\varsigma$ is a function on unknowns such that $\supp(\varsigma(X))\subseteq\atomsdown$ for every $X$.

Write $\Gamma\ment\varsigma$ when $X{:}\phi\in\Gamma$ implies $\varsigma(X)\in\idenot{\Gamma}{\phi}$ for every unknown $X$.
\end{defn}

\begin{rmrk}
Definition~\ref{defn.holes.valuation} seems harmless, but it carries some real meaning.
By condition~\ref{item.subset.condition} of Definitions~\ref{defn.lambda.model} and~\ref{defn.holes.lambda.model}, if $x\in\idenot{\Gamma}{\phi}$ then $\supp(x)\subseteq\dom(\Gamma)$.
This implies that if $X{:}\phi\in\Gamma$ then $X$ ranges over elements \emph{with support in $\dom(\Gamma)$}.
What happens to all the atoms in $\atomsdown\setminus\dom(\Gamma)$?
They cannot be used (unless we weaken the context with more typings).

This is related to a celebrated topic of continuing debate in the philosophy of language that assertions like `the King of France is bald' name and assert properties of apparently non-existent objects; they have meaning but do not denote \cite{russell:ond}.
In the same way, the variable $X$ asserts a property of all atoms in $\atomsdown$---that they may appear in the denotation of $X$---but this does not imply that these atoms exist in the possible world determined by the typing $\Gamma$.
The typing context determines which of the atoms in $\atomsdown$ have \emph{existential import} \cite{lambert:exiir}.
The extra twist to this story here, is that in nominal techniques atoms name themselves. 

This is another way of looking at the fine detail of the rule \rulefont{Meta}, that $b_i{:}\psi_i\in\Gamma$ even though $b_i\not\in\fa(X[b_i{\ssm}s_i]_1^n])$ in general.
In order to be substituted for, the atom $b_i$ must exist, and to exist it must be typed. 
\end{rmrk}

\begin{defn}
\label{defn.hole.interp}
Suppose $\Gamma\ment\varsigma$ and $\Gamma\cent r:\phi$. 
Define an \deffont{interpretation function} mapping $r$ to $\idenot{\varsigma;\Gamma}{r}$, by induction on $r$:
$$
\begin{array}{r@{\ }l@{\quad}l@{\qquad}r@{\ }l}
\idenot{\varsigma;\Gamma}{a}=&a^\iden_\phi &(a{:}\phi\in\Gamma)
&
\idenot{\varsigma;\Gamma}{C}=&C^\iden
\\
\idenot{\varsigma;\Gamma}{\lam{a{:}\phi}r}=&[a{:}\phi]\idenot{\varsigma;\Gamma}{r}  &(a\in\atomsup\setminus\dom(\Gamma))
&
\idenot{\varsigma}{rs}=&\idenot{\varsigma;\Gamma}{r}\bullet\idenot{\varsigma;\Gamma}{s}
\\
\idenot{\varsigma;\Gamma}{X\,[b_i\ssm s_i]_i}=&\varsigma(X)\,[b_i\sm\idenot{\varsigma;\Gamma}{s_i}]_i 
&(X{:}\phi\in\Gamma)
\end{array}
$$
\end{defn}

\noindent A few brief words on the case of $\lam{a{:}\phi}r$:
The condition $a\not\in\dom(\Gamma)$ prevents $a{:}\phi$ from overwriting typing information in $\Gamma$.
The condition $a\not\in\atomsdown$ ensures that the clause is well-defined, since otherwise $a$ might `accidentally capture' an atom in $\varsigma(X)$ for $X\in\fv(r)$.
The effect of $a$ capturing an atom in $X$ can be attained e.g. as $\lam{a{:}\phi}(X[a'{\ssm} a])$ where $a'\in\atomsdown$.

The language with holes (Definition~\ref{defn.holes.sorts}) is more expressive than the language without it (Definition~\ref{defn.pts.sorts}).
For instance, $a{:}\phi\ment r[a{\sm}a]=r$ (Definition~\ref{defn.ment}) is true for $r$ without unknowns, but otherwise may be false.
This is because without unknowns, we can use \rulefont{Suba} to \rulefont{Sub{\text{$\lambda$}}} to push substitution down to the atoms until it either vanishes or substitutes.
With unknowns this cannot be done; we may get `stuck' on a moderated unknown.

Put another way, $X$ really does range over arbitrary elements of the model whereas $a$ can only be \emph{substituted} for an arbitrary element of the model---and these are two distinct concepts.

\begin{xmpl}
\label{xmpl.hole.power}
Consider one base type and no constants and a nominal model $\interp I$ such that $\idenot{a{:}\tau}{\tau}=\{a^\iden_\tau,0,1\}$, where $\supp(0)=\supp(1)=\varnothing$.
Set $a^\iden_\tau[a{\sm} x]=x$, $0[a{\sm} x]=0$, and $1[a{\sm} x]=0$.

In STLC we cannot detect the element $1$ and its sensitivity to $[a{\sm} x]$ even though $a\not\in\supp(1)$.
In STLC extended with unknowns, we can.
Thus, we use \rulefont{Sub\#} instead of a weaker axiom that $b[a{\sm}x]=x$.
\end{xmpl}

\begin{thrm}[\bf First soundness theorem]
\label{thrm.hole.type.soundness}
If $\Gamma\ment\varsigma$ and $\Gamma\cent r:\phi$ then $\idenot{\varsigma;\Gamma}{r}\in\idenot{\Gamma}{\phi}$.
\end{thrm}
\begin{proof}
We consider two cases; the rest is as proof of Theorem~\ref{thrm.type.soundness}:
\begin{itemize*}
\item \emph{The case of \rulefont{L}.}\quad
Suppose $\Gamma,a{:}\phi\cent r:\psi$ and $a\in\atomsup$ so that by \rulefont{L} $\Gamma\cent\lam{a{:}\phi}r:\phi\fto\psi$.
By inductive hypothesis $\idenot{\varsigma;\Gamma,a{:}\phi}{r}\in\idenot{\Gamma,a{:}\phi}{\psi}$.
It follows from condition~3 of Definition~\ref{defn.holes.lambda.model}
that $[a{:}\phi]\idenot{\varsigma;\Gamma,a{:}\phi}{r}\in\idenot{\Gamma}{\phi\fto\psi}$.
By Definition~\ref{defn.hole.interp},\ $\idenot{\varsigma;\Gamma}{\lam{a{:}\phi}r}\in\idenot{\Gamma}{\phi\fto\psi}$.
\item \emph{The case of \rulefont{Meta}.}\quad
Suppose $\Gamma\cent s_i:\psi_i$ and $X{:}\phi\in\Gamma$ and $b_i{:}\psi_i\in\Gamma$ for $1{\leq}i{\leq}n$, so that by \rulefont{Meta} $\Gamma\cent X\,[b_i{\ssm}s_i]_i:\phi$.

By inductive hypothesis $\idenot{\varsigma;\Gamma}{s_i}\in\idenot{\Gamma}{\psi_i}$ and by assumption $\varsigma(X)\in\idenot{\Gamma}{\phi}$.
It follows by Lemma~\ref{lemm.holes.easy} that $\varsigma(X)\,[b_i{\ssm}\idenot{\varsigma;\Gamma}{s_i}]_i\in\idenot{\Gamma}{\phi}$. 
\qedhere\end{itemize*}
\end{proof}

\subsection{Soundness for $\beta$-conversion}

%

\begin{lemm}
\label{lemm.holes.beta.1}
Suppose $a_i{:}\phi_i\in\Gamma$ for $i\in A$ and $\Gamma\cent r:\psi$ and $\Gamma\cent s_i:\phi_i$ for $i\in A$.
Suppose $\Gamma\ment\varsigma$.
Then $\idenot{\varsigma;\Gamma}{r[a_i\ssm s_i]_{i{\in} A}}=\idenot{\varsigma;\Gamma}{r}[a_i\sm\idenot{\varsigma;\Gamma}{s_i}]_{i{\in}A}$. 
\end{lemm}
\begin{proof}
By a routine induction on the derivation of $\Gamma\cent r:\psi$.
We consider two cases:
\begin{itemize*}
\item \emph{The case of \rulefont{L} for $\lam{c{:}\chi}r$.} \quad
Renaming if necessary, suppose $c$ is fresh (so that $c\not\in\bigcup_{i{\in}A}(\supp(\idenot{\varsigma;\Gamma}{s_i})\cup\fa(s_i))\cup\dom(\Gamma)$).
We reason as follows:
$$
\begin{array}{r@{\ }l@{\qquad}l}
\idenot{\varsigma;\Gamma}{\lam{c{:}\chi}r}[a_i\sm\idenot{\varsigma;\Gamma}{s_i}]_{i{\in}A}=&
([c{:}\chi]\idenot{\varsigma;\Gamma}{r})[a_i\sm\idenot{\varsigma;\Gamma}{s_i}]_{i{\in}A}
&\text{Definition~\ref{defn.hole.interp}}
\\
=&
[c{:}\chi](\idenot{\varsigma;\Gamma}{r}[a_i\sm\idenot{\varsigma;\Gamma}{s_i}]_{i{\in}A})
&\rulefont{Sub\text{$\lambda$}},\ c\not\in\supp(\idenot{\varsigma;\Gamma}{s})
\\
=&
[c{:}\chi]\idenot{\varsigma;\Gamma}{r[a_i\ssm s_i]_{i{\in}A}}
&\text{ind. hyp.}
\\
=&
\idenot{\varsigma;\Gamma}{\lam{c{:}\chi}(r[a_i\ssm s_i]_{i{\in}A})}
&\text{Definition~\ref{defn.hole.interp}}
\\
=&
\idenot{\varsigma;\Gamma}{(\lam{c{:}\chi}r)[a_i\ssm s_i]_{i{\in}A}}
&c\not\in\fa(s)
\end{array}
$$
\item
\emph{The case of \rulefont{Meta}.}\quad
We reason as follows, where $B=\{1,\dots,n\}$ and $\{b_j\mid j\in B\}\subseteq\atomsdown$:
$$
\hspace{-3.5em}\begin{array}{r@{\ }l@{\quad}l}
\idenot{\varsigma;\Gamma}{X[a_j{\ssm}t_j]_{j{\in}B}}[a_i{\sm}\idenot{\varsigma;\Gamma}{s_i}]_{i{\in}A}
=&
\varsigma(X)[a_j{\sm}\idenot{\varsigma;\Gamma}{t_j}]_{i{\in}B}[a_i{\sm}\idenot{\varsigma;\Gamma}{s_i}]_{i{\in}A}
&\text{Defn.~\ref{defn.hole.interp}}
\\
=&
\varsigma(X)[(a_i{\sm}\idenot{\varsigma;\Gamma}{s_i})_{i{\in}A{\setminus}B,\ a_i{\in}\atomsdown},(a_j{\sm}\idenot{\varsigma;\Gamma}{t_j}[a_i{\sm}\idenot{\varsigma;\Gamma}{s_i}]_{i{\in}A})_{j{\in}B}]
&\text{fact}
\\
=&
\varsigma(X)[(a_i{\sm}\idenot{\varsigma;\Gamma}{s_i})_{i{\in}A{\setminus}B,\ a_i{\in}\atomsdown},(a_j{\sm}\idenot{\varsigma;\Gamma}{t_j[a_i{\ssm}s_i]_{i{\in}A}})_{j{\in}B}]
&\text{ind. hyp.}
\\
=&
\idenot{\varsigma;\Gamma}{X[(a_i{\ssm}s_i)_{i{\in}A{\setminus}B,\ a_i{\in}\atomsdown},(a_j{\ssm}t_j[a_i{\ssm}s_i]_{i{\in}A})_{j{\in}B}]}
&\text{Defn.~\ref{defn.hole.interp}}
\\
=&
\idenot{\varsigma;\Gamma,a{:}\phi}{X[a_j{\ssm}t_j]_{j{\in}B}[a_i{\ssm}s_i]_{i{\in}A}}
&\text{Defn.~\ref{defn.holes.sub}}
\end{array}
$$
\end{itemize*}
Some detailed calculations are hidden in the `fact' used above.
This follows using \rulefont{Sub\text{$\lambda$}} and \rulefont{SubApp} from Definition~\ref{defn.sim.sub}, and is one reason that in that definition we `freshened' the $b_i$ to $c_i$; to avoid clash.
\end{proof}

\begin{lemm}
\label{lemm.holes.beta.2}
Write `$\Gamma\cent r,s:\phi$' as shorthand for `$\Gamma\cent r:\phi$ and $\Gamma\cent s:\phi$'.
Suppose $\Gamma\ment\varsigma$.
\begin{enumerate*}
\item
Suppose $\Gamma\cent r,r':\phi\fto\phi'$ and $\Gamma\cent s,s':\phi$. 
If $\idenot{\varsigma;\Gamma}{r}{=}\idenot{\varsigma;\Gamma}{r'}$ and $\idenot{\varsigma;\Gamma}{s}{=}\idenot{\varsigma;\Gamma}{s'}$ then $\idenot{\varsigma;\Gamma}{rs}{=}\idenot{\varsigma;\Gamma}{r's'}$. 
\item
Suppose $\Gamma,a{:}\phi\cent r,r':\psi$.
If $\idenot{\varsigma;\Gamma,a{:}\phi}{r}=\idenot{\varsigma;\Gamma,a{:}\phi}{r'}$ then $\idenot{\varsigma;\Gamma}{\lam{a{:}\phi}r}=\idenot{\varsigma;\Gamma}{\lam{a{:}\phi}r'}$.
\item
Suppose $X{:}\phi\in\Gamma$ and $\{b_i{:}\phi_i\mid 1{\leq}i{\leq}n\}\subseteq\Gamma$.
Suppose $\Gamma\cent s_j,s'_j:\psi_j$ and $\idenot{\varsigma;\Gamma}{s_j}=\idenot{\varsigma;\Gamma}{s'_j}$ for $1{\leq}j{\leq}n$.
Then $\idenot{\varsigma;\Gamma}{X[b_i{\ssm}s_i]_1^n}=\idenot{\varsigma;\Gamma}{X[b_i{\ssm}s'_i]_1^n}$.

\end{enumerate*}
\end{lemm}
\begin{proof}
These are facts of equality in sets.
\end{proof}

\begin{corr}[\bf Second soundness theorem]
If $r\beq s$ (Defn.~\ref{defn.holes.beta.equivalence}) and $\Gamma\cent r:\phi$ and $\Gamma\models\zeta$ then $\idenot{\varsigma;\Gamma}{r}=\idenot{\varsigma;\Gamma}{s}$.
\end{corr}
\begin{proof}
Using Lemmas~\ref{lemm.holes.beta.1} and~\ref{lemm.holes.beta.2}.
\end{proof}

We would expect a completeness result like Theorem~\ref{thrm.st.complete} to hold and have a similar proof.
Details will be in a journal version.
We should verify soundness under instantiating unknowns:
 

\begin{thrm}[\bf Third soundness theorem]
Suppose $\Gamma,X{:}\chi\cent r:\phi$ and $\Gamma\cent t:\psi$.
Suppose $\Gamma\ment\varsigma$.

Then $\idenot{\varsigma;\Gamma}{r[X{\ssm} t]}=\idenot{\varsigma[X{\ssm}\idenot{\varsigma;\Gamma}{t}];\Gamma}{r}$.
\end{thrm}
\begin{proof}
We consider \rulefont{Meta} for $X$.
Suppose $\Gamma=\Gamma',X{:}\chi,(b_j{:}\psi_j)_1^n$ and $\Gamma\cent s_i{:}\psi_i$ for $1{\leq}i{\leq}n$ so that by \rulefont{Meta} $\Gamma',X{:}\chi,(b_j{:}\psi_j)_1^n\cent X\,[b_i{\ssm}s_i]_i$.
Then
$
\idenot{\varsigma[X\ssm \idenot{\varsigma;\Gamma}{t}];\Gamma}{X\,[b_i{\ssm}s_i]_i}
\stackrel{\text{Defn.~\ref{defn.hole.interp}}}{=}
\idenot{\varsigma;\Gamma}{t}[b_i{\sm}\idenot{\varsigma;\Gamma}{s_i}]_i
\stackrel{\text{Lemm.~\ref{lemm.holes.beta.1}}}{=}
\idenot{\varsigma;\Gamma}{t[b_i{\ssm}s_i]_i} . 
$
\end{proof}

\section{Conclusions}

We have built a semantics for the simply-typed $\lambda$-calculus (STLC), based on nominal sets.
In keeping with the `nominal' philosophy, variables (names) are denoted by themselves. 
This draws certain structure that is normally external to the denotation inside it, and this extra structure excludes some arguably pathological homomorphisms between models. 
We also exploit the semantics to existential variables, or `holes' (suggesting that we do not just get more models out of this nominal semantics, but also more languages).

The constructions are not really any harder than for traditional STLC semantics.
When reading for instance Definition~\ref{defn.lambda.model}, the reader should mentally place this side-by-side with a \emph{full} specification of traditional STLC semantics, including for instance a precise definition of valuations as graphs (these are functions with the general shape $(\mathbb A\to X)\to X$).
Our nominal semantics for STLC is no harder than what the reader already knows; it is just different.

Atypically for nominal techniques so far as exemplified e.g.\ by \cite{gabbay:nomu-jv,gabbay:frepbm,cheney:alppl,gabbay:pernl}, atoms have non-trivial types. 
These, if they assign atoms any type information at all, assign them `the type of atoms'.
There has been some work assigning more interesting types to atoms \cite{gabbay:curstn}, but not in denotations.

In the course of doing all these things, we note echoes of other strands of research.
The distinction between $b\in\atomsdown$ and $b:\phi\in\Gamma$ is an instance of the distinction between meaning and denotation (only finitely many of the atoms in the permission set of an unknown have existential import in the denotation) \cite{russell:ond,lambert:exiir}.
Our use of $\atomsdown$ and $\atomsup$, which is borrowed from \cite{gabbay:perntu,gabbay:perntu-jv}, is reminiscent of the two kinds of variable used by Frege \cite{heijenoort:fregsb} (for a more modern presentation see e.g. \cite[Chapter~IV, Section~1]{smullyan:firol}).
This is more an analogy than a precise correspondence and we will discuss matters further in a longer paper where we have more space to develop the syntax.
We still have only one set of atoms and the `nominal' constructions, notably the notions of support, binding, freshness, and nominal set, are unchanged.
We have freely imported ideas from (permissive) nominal terms, notably in our treatment of existential variables.

\subsection{Related work}

\paragraph{Valuations and unknowns.} 
We gave unknowns a semantics using valuations in Definition~\ref{defn.hole.interp}.
Arguably it is disappointing: why map atoms to themselves in a denotation but then switch to another (more traditional) methodology for unknowns?
One answer is that atoms are universal variables (could be replaced by anything) whereas unknowns are existential variables (must be replaced by something), so it is reasonable to interpret them using a valuation, and perhaps we should. 
Indeed there is a precedent for this: atoms correspond to \emph{$\delta$-variables} and unknowns to \emph{$\gamma$-variables} from \cite{Wirth:desid}.
Making this formal by considering a paper similar to this one but aimed at first-order logic is a topic of current research.

Still, there is an interesting alternative.
In \cite{gabbay:twolns} a direct nominal semantics is explored for unknowns $X$, analogous to how atoms $a$ map to themselves in this paper, called \emph{two-level nominal sets}.
Two-level nominal sets \emph{with substitutions} would provide a theory of incompleteness in which holes are directly represented in the semantics.
There is no need for that in this paper because we have no level 3 variables (in the style e.g. of the $\lambda$-context calculus \cite{gabbay:lamcce}); but if there were, two-level nominal sets might be not only interesting, but necessary. 

\paragraph{Models as presheafs.}
The reader familiar with category theory will recognise in Definition~\ref{defn.lambda.model} a presheaf.
In fact, we have enriched the usual presheaf $\f{Sets}^{\mathbb I}$ (presheaves over finite sets and injections between them) to a presheaf over an indexing category enriched with types.
Condition~5 of the two definitions (for $\cap$) states that these presheaves should preserve pullbacks of monos; this is the critical property required for the sets-based presentation of this paper to work \cite{gabbay:nomrs}.
Presheaves enriched over types have appeared in \cite{zsido:typas}, without the nominal sets style presentation and written for a different audience (one stemming from view of syntax and substitution based on \cite{fiore:abssvb}).
The presheaves are used differently: by considering initial objects, inductive datatypes of \emph{well-typed} syntax-with-binding are constructed.

\paragraph{Other theories of functions.}
\emph{Combinatory algebra} (CA) assumes constants $S$, $K$, and $I$.  Axioms allow them to model the $\lambda$-calculus.
However, CA is strictly weaker than the theory of $\beta$-conversion; the \rulefont{\xi} rule cannot be equationally axiomatised, because $\lambda$ cannot be directly expressed (though any given $\lambda$-term can be compiled to combinators).
This can be fixed using explicit indeterminates which, from the point of view of this paper, look a lot like atoms \cite{selinger:lamca}.

Alternatively, \emph{lambda-abstraction algebras} (LAAs) are a first-order axiomatisation which does satisfy $\xi$ \cite{salibra:appual}; again, from the point of view of this paper LAAs look much like the axioms we have considered.
LAAs are not typed; a `nominal' equivalent of them was considered by the first author with Mathijssen \cite{gabbay:nomalc}.  So this paper is significantly different from both since, as we see comparing this paper with \cite{gabbay:nomalc}, the addition of types makes a real difference to the models.
 LAAs take semantics in `ordinary' sets, so their semantics is not well-pointed in the sense of this paper, and we do not obtain the language with meta-variables which we have developed here or relate so directly with a wider research context (e.g. into nominal techniques).   One further subtle feature of LAA models is that they do not have finite support (we do not argue whether this is good or bad; we merely observe this as a significant difference).  

Salibra and others have thought deeply about the lattice properties of $\lambda$-calculus models.  As a final note we mention that nominal algebra satisfies an HSPA theorem \cite{gabbay:nomahs}, and \emph{permissive-}nominal algebra satisfies an HSP theorem \cite{gabbay:nomtnl}.  This has also been considered by Kurz and others \cite{kurz:unians}.  The deeper theory here---how theorems of universal algebra applied to $\lambda$-calculus adapt to the nominal context---remains unexplored. 

\paragraph{Other theories of existential variables.}
In implemented systems like LF and Isabelle \cite{PaulsonLC:fougtp} these are handled as a special syntactic category of higher-order variable.
That is, an unknown of type $\iota$ depending on (universal) variables of type $\tau$ and $\tau'$ is modelled by a variable of type $\tau\to\tau'\to\iota$.
We make no claim that our model of existential variables is better in implementation---it is simply too early to tell---but generally speaking we are against solving problems by moving to higher orders. 
Plenty of complexity can be \emph{encoded} in function spaces, and this is fine for implementation, but encoding something is not the same as having a good mathematical model of it. 
Jojgov includes a excellent and detailed discussion of this issue in \cite{jojgov:holbp}, which is his own analysis of incompleteness;
intuitively, by conflating $\beta$-conversion with incompleteness it becomes impossible to distinguish between a complete derivation of higher type, and an incomplete derivation of lower type.
We add that the denotation of $\tau\to\tau'\to\iota$ is uncountable, even if the denotations of $\tau$, $\tau'$, and $\iota$ are countable.
We would not immediately expect there to be uncountably many existential variables of type $\iota$, so if only on the grounds of size we would hope for something smaller.
Our denotations deliver this: an existential variable of type $\iota$ is just an unknown $X:\iota$.

Contextual modal type theory (CMTT) has two levels of variable; it enriches STLC with `modal types' representing open code \cite{nanevski:conmtt}.
However, level 2 variables of CMTT are not existential variables; they are a species of \emph{intensional} variable ranging over code.
Making this formal using a nominal semantics related to the semantics of this paper, is current research by the first author.

\subsection{Future work}

We note that, as it stands, there is no general mathematical framework for the study of incomplete terms in type theory.
Implementors of proof assistants invent \emph{ad hoc} methods for representing incomplete terms, representing incomplete proof states, in their systems.
Methods evolve more through trial and error than deep mathematical insight.
For instance, early versions of Coq used a complex system involving two syntactic classes---`existential variables' and `metavariables'---for representing incomplete proof states.
Matita~\cite{asperti:crafting:2007}, whose design was influenced by lessons learned in Coq's development, used from the outset a much simpler scheme where the concepts of `existential variable' and `metavariable' are unified~\cite{sacerdoti:thesis}.

We hope that the work presented in Section~\ref{sect.existential.variables} forms the basis for further, mathematical study of incomplete terms in type theory. 
For instance, the model of STLC in Definition~\ref{defn.lambda.model} could be extended to a dependent type theory like e.g. compact $\lambda P$ \cite[Subsection~14.2, Figure~14.1]{sorensen:lecchi-book} (with or without incompleteness).

It should be fairly easy to internalise the substitution action for unknowns by adding $\lambda X$, thus obtaining a two-level system with logic and computation at both levels---the result should resemble the first two levels of the lambda-context calculus \cite{gabbay:lamcce}, but with a stronger theory of $\alpha$-equivalence and more reductions. 
One concrete application of this may be to expressing tactics---functions from incomplete derivations to incomplete derivations---in type-theory based theorem-provers, which need to program on terms (considered as computation or proof respectively).  More goes into such a design than metavariables, but the character of metavariables is key to that of existing implementations \cite{pientka:belfpr}.

Notions of incompleteness can be motivated by efficiency and speed; notably \cite{pientka:opthop} was motivated by optimising unification in LF.
These ideas have led to several implementations; an up-to-date overview is in \cite{pientka:belfpr}.
Note that the details of the syntax are different: the work uses a two-level type system with special types for closed code, and `meta-variables' range over \emph{closed} elements of the domain (i.e. $\supp(\varsigma(X))=\varnothing$, intuitively).
No general semantic theory has been given for this line of research, and we suspect that the nominal denotations of this paper could be turned to that task.

\hyphenation{Mathe-ma-ti-sche}

\end{document}